\documentclass[conference]{IEEEtran}
\IEEEoverridecommandlockouts
\usepackage{lipsum}
\usepackage{graphicx}
\usepackage{amsmath,amssymb,amsfonts}
\usepackage{hyperref}
\usepackage[utf8]{inputenc}
\usepackage{cite}
\usepackage{algorithm}
\usepackage{graphicx}
\usepackage{textcomp}
\usepackage{breqn}
\usepackage{url}
\usepackage{epsfig}
\usepackage{float}
\usepackage{color}
\usepackage{balance}
\usepackage{bbm}
\usepackage{amsthm}
\usepackage{textcomp}
\usepackage{enumitem}
\usepackage{tikz}
\usetikzlibrary{automata, positioning}
\usepackage{amsmath}
\usepackage{pgfplots}
\usepackage{subcaption}
\usepackage{graphicx}
\usepackage{animate}
\usepackage{tikz}
\usepackage{pgfplots}
\usepackage{algpseudocode}
\usepackage{tikz}
\usetikzlibrary{positioning,arrows.meta}

\usepackage{mathtools}
\usepackage{pgffor}
\usepackage{pgfplots}
\usepackage{pgfplotstable}
\usepackage{booktabs}
\usepackage{algpseudocode}
\usepackage[font=small,labelfont=bf]{caption}

\newtheorem{theorem}{Theorem}
\newtheorem{definition}{Definition}
\DeclareMathOperator{\Max}{Max}

\hypersetup{
    colorlinks=true,
    linkcolor=black,   
    citecolor=blue,    
    urlcolor=blue      
}

\algnewcommand\algorithmicforeach{\textbf{for each}}
\algdef{S}[FOR]{ForEach}[1]{\algorithmicforeach\ #1\ \algorithmicdo}

\renewcommand\algorithmicdo{}

\begin{document}

\title{Adaptive Learning for Moving Target defence: Enhancing Cybersecurity Strategies}

\author{Mandar Datar\IEEEauthorrefmark{1}\thanks{$^{1}$ M. Datar is currently with 
CEA-Leti, Grenoble, France
         and was associated with Orange Research, Châtillon, France, during the time of this work. },  Yann Dujardin\IEEEauthorrefmark{2}
\\
\IEEEauthorrefmark{1}CEA-Leti, Grenoble, France
\IEEEauthorrefmark{2} Orange Research, Châtillon, France

}

\maketitle

\begin{abstract}
In this work, we model Moving Target Defence (MTD) as a partially observable stochastic game between an attacker and a defender. The attacker tries to compromise the system through probing actions, while the defender minimizes the risk by reimaging the system, balancing between performance cost and security level. We demonstrate that the optimal strategies for both players follow a threshold structure. Based on this insight, we propose a structure-aware policy gradient reinforcement learning algorithm that helps both players converge to the Nash equilibrium. This approach enhances the defender's ability to adapt and effectively counter evolving threats, improving the overall security of the system. Finally, we validate the proposed method through numerical simulations.
\end{abstract}

\section{Introduction}
Cybersecurity consists in preventing cyberattacks, i.e. attempts by hackers to, for example, damage, destroy a system or a server, lock systems to ask money for unlocking (ransomware), or to get access to sensitive data (financial, personal, confidential, ...). Among others. With the growing number of internet users and increasing network traffic, automated defence mechanisms have become more crucial than ever. Traditional cybersecurity relies on a detection and response approach, where systems identify early signs of threats and take action.  Despite many advances, computer systems and networks remain vulnerable to various cyberattacks over time.

\par A key challenge in achieving effective defence lies in information asymmetry, where attackers typically possess more knowledge about the defender than the defender has about the attacker. Over time, attackers can engage in continuous reconnaissance efforts to gather information, analyze the system, and study the target's defence policies.  This allows them to adjust and improve their attack strategies continuously.

To counter this information asymmetry and overcome traditional defence methods, moving target defence (MTD) has evolved as a promising solution, with increasing research and deployment interest in recent years [1-10]. This cybersecurity approach dynamically reconfigures the system, for instance, through virtual machine migration and IP shuffling to alter the attack surface. Essentially, it ``moves" cyber assets (or targets) within an ICT network, making it harder for attackers to exploit vulnerabilities \cite{lei2018moving}.

\par Although MTD has emerged as a promising strategy for enhancing cybersecurity, it also introduces new challenges. The system reconfiguration is inherently costly, particularly in terms of performance. Frequent reconfigurations can lead to a degradation in performance, resulting in increased delays and a reduction in service quality. In contrast, infrequent reconfigurations  can expose the system to greater risk. Therefore, it is essential to optimize the reconfiguration frequency in order to achieve a balance between system performance and security.

\par In recent years, researchers have studied MTD from various perspectives and have proposed a range of reconfiguration solutions \cite{sengupta2020survey}. The authors of \cite{dass2021reinforcement} proposed a method based on reinforcement learning (RL) to generate secure software configurations for general MTD. The authors of \cite{yoon2021desolater} used multi-agent Deep Reinforcement Learning (DRL) to implement proactive IP shuffling for MTD  against reconnaissance attacks in in-vehicle SDNs. Several other studies have explored RL based MTD method, including \cite{celdran2024rl,shang2023deep,lallouche2024deep}. Many studies have modeled MTD as a game between the defender and the attacker, and analyzed it through the lens of game theory. For instance, \cite{ eghtesad2020adversarial} introduced a method that models the problem as a two-player general-sum game, where an attacker and a defender compete to control a group of servers. To find the optimal MTD strategy, they utilized Deep Q-Learning (DQ-Learning) and the Double Oracle Algorithm. In \cite{sengupta2020multi}, the authors used Bayesian Stackelberg Markov Games (BSMGs) to model uncertainty about attacker types and MTD specifics, and proposed a multi-agent RL-based solution to learn optimal movement policies.
Authors in \cite{feng2017stackelberg} designed optimal timing of MTD based on the stackelberg game model.

In this work, we propose a novel game-theoretic approach to optimal reconfiguration or reimage under MTD to balance costs and security risks effectively. In particular, we consider the defender can reimage the network according to partial observations of the attacker's probings, and the attacker makes actions according to perfect observations of the system (he is aware of who controls the system at every time step). The framework is more realistic as it corresponds to the situation when talking about the attack/defence process, ii) it allows us to look at theoretical guarantees and bounds about the policies search, iii) these results can be observed numerically.
\par Our work is close to ~\cite{ eghtesad2020adversarial}, we model the MTD as a partially observable stochastic game between the attacker and the defender. However, our work departs from there in that they directly used multi-agent DRL to learn the Nash equilibrium policy; while we investigate the Nash equilibrium policy and show that the Nash equilibrium policy has a threshold structure over belief of state, a similar approach has been used in \cite{hammar2023learning}  but for automated intrusion response. Finally, relying on structural information, we propose a policy gradient reinforcement learning algorithm to compute the Nash equilibrium. Our approach method significantly reduces the complexity of Nash equilibrium computation.


The remainder of the paper is organized as follows. Section \ref{sysytem model} describes the system model. In Section \ref{Problem_formulation}, we analyze the Nash equilibrium within the framework of a partially observable stochastic game. Section \ref{sec_learning} introduces a policy gradient-based reinforcement learning algorithm to compute the equilibrium. Numerical results evaluating the performance of the proposed approach are presented in Section \ref{numerical}. Finally, we summarize our findings and discuss future research directions. 
\section{Model}\label{sysytem model} We consider a single system that typically represents a critical cyber system.
\subsection{Attacker}
We consider attacker attempts to gain control over a system through a series of probing actions. With each probe, he either can succeed in gaining control of the system or it increases the attacker's chances of gaining control in subsequent attempts. Such aspects have been generally observed in practicality as the attacker gathers more information at each stage, and his understanding of the target's infrastructure, security measures, and potential vulnerabilities increases. 
Attacker can probe system, probing system takes control of it with probability 
\begin{equation}
1-e^{-\alpha(\rho+1)}
\end{equation} where $\rho$ is the number of probes attempted. 
\subsection{Defender} In order to protect the system, we consider the defender adopts the practice of periodically reimaging it, which involves reconfiguring the system to its original state. This proactive measure hinders attackers from establishing a strong presence within the system. Any progress made by the attackers in infiltrating the system will be completely wiped out during the reimage process. However, reimagining a system can bring about various costs for the defender. For instance, changing configurations or network topologies too often may impact system performance and responsiveness, and it can also result in the dissatisfaction of customers due to delays. Thus the defender must strike the right balance between enhanced security measures and potential performance degradation. We model the system as a partially observable noncooperative stochastic game, where two players, Defender and Attacker, compete to gain control over a system or a security-sensitive resource. The one who tries to compromise the system is called an attacker, while the defender is the one who protects the system from being compromised. Formally, we define the game as consisting of the following elements
\subsection{Players}
\begin{itemize}
\item Defender (D)

\item Attacker (A)
\end{itemize}

\subsection{Action set}
Let $\mathcal{A}_D$ and $\mathcal{A}_A$ denote the action sets of the defender and the attacker, respectively. The defender can reimage system\\
$\mathcal{A}_D=\left\{0,1 \right\}$, where 0:=reimage, 1:=continue\\
The attacker can probe the system and seek to gain control over it.
$\mathcal{A}_A=\left\{0,1 \right\}, $ where 0:=probe, 1:= not probe 
\subsection{Rewards and Costs}
\begin{itemize}
    \item Defender's Cost: Reimaging the system incurs a cost \( C_D \).
    \item Attacker's Cost: Each probing attempt costs \( C_A \).
    \item Control Reward: A player gains 1 unit of reward per time step if they control the system.
\end{itemize}
\subsection{State Space}
We define the states of the system as follows.
$\mathcal{S}= \left\{0,1 \right\}$
\begin{itemize}
    \item[0] defender controls the system
    \item[1] attacker controls the system
\end{itemize}
\subsection{State Transition Probabilities}
Let $\mathcal{T}(j, i, d, a)$ represent the probability of transitioning from state $i$ to state $j$ when the defender takes action $d$ and the attacker takes action $a$. Before proceeding, it is important to note that the in our case system exhibits the memoryless property. This means that the probability of success at a future probe is independent of any previous failures, which can be verified by the following calculation.
\begin{dmath}
P(\text{Success at probe } \rho+k \mid \text{Failure up to } \rho) = P(\text{Success at probe } k)
\end{dmath}
\begin{itemize} \item If the defender controls the system (state 0) and decides not to reimage (action $d = 1$), while the attacker probes (action $a = 0$), the system stays in state 0 with probability $e^{-\alpha(\rho+1)}$. The remaining probability, $1 - e^{-\alpha(\rho+1)}$, represents the case where the attacker gains control and the system moves to state 1. \item If the defender decides to reimage the system, regardless of the current state, the system will always reset to state 0 (i.e., defender controls the system). \item If the defender does not reimage (action $d = 1$) and the attacker does not probe (action $a = 0$), the system stays in the same state with probability 1. This means if the defender controls the system (state 0), it stays in state 0, and if the attacker controls the system (state 1), it stays in state 1. \item Finally, if the defender decides not to reimage and the attacker controls the system, the system will remain in state 1 regardless of the attacker's action. \end{itemize}
The transition probabilities can be summarized as follows:
\begin{align} \mathcal{T}(0, 0, 1, 0) &= e^{-\alpha} \\
\mathcal{T}(1, 0, 1, 0) &= 1 - e^{-\alpha} \\
\mathcal{T}(0, \cdot, 0, \cdot) &= 1 \\
\mathcal{T}(0, 0, 1, 1) &= 1 \\ 
\mathcal{T}(1, 1, 1, .) &= 1
\end{align}\subsection{Observations}
The defender does not know whether system has been compromised by the the attacker or not. However, the defender can observe the each probe with probability $1-\nu$ hence with probability $\nu$ probe is undetected. The attacker comes to know if defender reimage system if it is compromised by the him, but cannot observe the reimaging of uncompromised system without probing it  
\subsection{Belief Space}
Since the defender does not know the exact state of the system, it can only detect the attacker's probing attempts with some probability. Based on the observed history, the defender forms belief about the underlying state of the game. Let \( b_k^D(s_k) \) denote the defender’s belief about the system state \( s_k \) at time \( k \), given the history \( h_k \) up to that time
\begin{equation}
    b_k^D = \mathbb{P}(s_k \mid h_k)
\end{equation}
\subsection{The observation distribution}
Let \( B(d) \) represent the observation distribution for the defender. It indicates the probability of observing event \( o \) given that the defender took action \( d \) and the system transitioned to state \( s \). Mathematically, this is expressed as:
\[
B_{so}(d) = \mathbb{P}\left(o_{k+1} = o \mid x_{k+1} = s, d_k = d\right)
\]
This means the probability that the observation \( o_{k+1} \) is equal to \( o \) at the next time step, given that the state at that time step is \( s \) and the defender took action \( d \) at time step \( k \).
\subsubsection{Belief Update with HMM Filter}
Let \( b^D_k = \left[b^D_k(0), b^D_k(1)\right] = \left[b^D_k(0), 1 - b^D_k(0)\right] \) represent the defender's belief about the state of the system at time \( k \), where \( b_k(0) \) is the belief that the system is controlled by the defender at time \( k \). The belief at the next time step \( k+1 \) is updated as follows:

\[
b_{k+1} = T_{\pi_A}(o_{k+1}, b_k, d_k)
\]

Here, \( T_{\pi_A}(o_{k+1}, b_k, d_k) \) represents the update of belief in step \( k+1 \), based on the prior belief \( b_k \) in step \( k \), the observation \( o_{k+1} \) made at step \( k+1 \), and the action \( d_k \) taken by the defender at step \( k \). Based on the Bayes formula $p(a_i/b)=\frac{p(b/a_i)p(a_i)}{\sum_j p(b/a_j)p(a_j)}$, the belief update can be expressed in terms of transition probabilities and observation distribution as follows. 
\begin{align}
T\left(y,\pi,u\right)=\frac{B_y(u)P(u)\pi}{\sigma(y,\pi,u)}\label{T_map}
\end{align}
where 
\begin{align}
\sigma(y,\pi,u)=\mathbf{1'}B_y(u)P(u)\pi
\end{align}
\subsection{Strategies}

Let \( \pi_{D} \in \Pi_{D} \) represent the defender's strategy, which maps the defender's current belief about the state to a probability distribution over the possible actions. This can be written as:

\[
\pi_D(b^D_t) \rightarrow \Delta(\mathcal{A}_{D})
\]

In our case, there are only two states and two possible actions. Therefore, the strategy can be viewed as a function that takes the defender's current belief (that the attacker controls the system) and returns the probability of taking action 0 (reimqge). Similarly, the attacker’s strategy maps from the current state and the belief of the defender to a probability distribution over the possible actions of the attacker. This is represented as:

\[
\pi_A(S, b^D_t) \rightarrow \Delta(\mathcal{A}_A)
\]
In short, this can be seen as the mapping from the state to the probability of taking action 0.
\subsubsection{Objective functions}
We consider both players to be rational and seek to maximize their infinite-horizon discounted expected rewards. Given the opponent's strategy, their respective objective functions are defined as:

\begin{align}
J_{D}(\pi_D,\pi_A) = \mathbb{E}_{(\pi_D,\pi_A)} \left[ \sum_{t=1}^{\infty} \gamma^{t-1} R_D(s_t,a_t,d_t) \right], \\
J_{A}(\pi_D,\pi_A) = \mathbb{E}_{(\pi_D,\pi_A)}  \left[ \sum_{t=1}^{\infty} \gamma^{t-1} R_A(s_t,a_t,d_t) \right].
\end{align}
where $\gamma$ is the discount factor, and $\mathbb{E}_{(\pi_D,\pi_A)}$ represents the expectation on the random variable $(S_t,A_t)$ generated by the policy pair $(\pi_D,\pi_A)$ 
\subsubsection{Nash Equilibrium}
We assume that both players want to maximize their reward. A policy pair $(\pi^*_D,\pi^*_A)$ is said to be Nash equilibrium if no player can benefit by deviating unilaterally, mathematically defined as follows. 

\begin{align}
    J_{D}(\pi^*_D,\pi^*_A)\geq J_{D}(\pi_D,\pi^*_A), \quad\forall(\pi_D^*,\pi_D)\in \Pi_D\\
J_{A}(\pi^*_D,\pi^*_A)\geq J_{A}(\pi^*_D,\pi_A), \quad\forall(\pi^*_A,\pi_A)\in \Pi_A
\end{align}
In the next section, we analyze the Nash equilibrium 

\section{Analyzing Nash Equilibrium}\label{Problem_formulation}

The Nash equilibrium in the attacker-defender game can be found by identifying strategy pairs where each player's strategy is the best response to the other's. The defender's best response is determined by solving a Partially Observable Markov Decision Process (POMDP), while the attacker's best response is obtained by solving a Markov Decision Process (MDP). Since a POMDP is a continuous-state MDP with state space being the unit simplex, we
can straightforwardly write down the dynamic programming equation for the optimal
policy as for continuous-state MDPs.

\begin{dmath}
    V_{\pi_A}(b)= \Max_{d\in \mathcal{A}_D}\mathbb{E}_{\pi_A}\left [ R_D(b,d) +\sum_{o\in \mathcal{O}} V_{\pi_A}(T\left ( o,b ,d \right ) )\sigma \left ( o,b ,d \right ) \right ]
\end{dmath}
For the attacker, 
\begin{dmath}
    V_{\pi_D}(s,b)= \Max_{d\in \mathcal{A}_A}\mathbb{E}_{\pi_D}\left [ R_A(s,a) +\sum_{s'\in \mathcal{S}} V_{\pi_D}( s'  )P \left (s',s,a  \right ) \right ]
\end{dmath}

In the next section, we analyze the structure of Nash equilibrium policies in the attacker-defender game. Before proceeding, we recall some key definitions needed for proving the theoretical results.
\begin{definition}[First-order stochastic dominance]
   Let $\mu_1, \mu_2 $ denote two pmfs Then $\mu_1$ is said to first-order
stochastically dominate $\mu_2$, $\mu_1\geq_s \mu_2$ if 
\begin{align}
\sum_{i=j}^{X}\mu_1(i)\geq \sum_{i=j}^{X}\mu_2(i)
\end{align}
\end{definition}

\begin{definition}[Monotone likelihood ratio (MLR) ordering]
    Let $b_1,b_2$
denote two belief state vectors. Then $b_1$ dominates $b_2$ with respect to the MLR order,
\begin{align}
    b_1(i)b_2(j)\leq b_2(i)b_1(j)\quad i<j \quad i,j \in \mathcal{X}
\end{align}
\end{definition}

For a state space of dimension $2$
\[b_1 \geq_r b_2  \Leftrightarrow b_1 \geq_s b_2 \Leftrightarrow b_1(2) \geq b_2(2) \] 

\begin{definition}[Totally positive of order 2 (TP2)] A transition or observation kernel denoted matrix $M$ is TP2 if the $(i + 1)^{th}$ row
MLR dominates the $i^{th}$ row: that is, $M_i \geq_r M_j$, for every $i>j$
\end{definition}
\begin{definition}[tail-sum supermodular]
    A transition probabilities  $P_{ij}(u)$ are tail-sum supermodular in $(i,u)$ if $\sum_{j\geq l} P_{ij}(u)$ is supermodular in $(i,u)$ \emph{i.e,}
    \begin{align*}
        \sum_{j=l}^{X}(P_{i,j}(u+1)-P_{i,j}(u)) \text{ is increasing in i,\;} i\in \mathcal{X},u\in \mathcal{U} 
\end{align*}
which can be written in terms of first order stochastic dominance as 

\begin{align*}
     \frac{1}{2}\left ( P_{i+1}(u+1)+ P_{i}(u)\right )\geq_s \frac{1}{2}\left ( P_{i}(u+1)+ P_{i+1}(u)\right )
\end{align*}
\end{definition}
We show that the nash equilibrium policy for the defender is threshold in its belief about the state of system 
\begin{theorem}
 
Given an attacker policy \(\pi_A \in \Pi_D\), the defender's value function \( V(b) \) is decreasing in the belief state \( b \). Moreover, the defender’s optimal policy follows a threshold structure, i.e., it is decreasing in \( b \).
\end{theorem}
\begin{proof}
    \begin{enumerate}
        \item First we show that defender's reward function is decreasing and sub-modular
\begin{align}
    \begin{matrix}
R_D\left(0,0\right)=1-C_D & R_D\left(0,1\right)= 1  \\
R_D\left(1,0\right)=1-C_D & R_D\left(1,1\right)=0\\
\end{matrix}
\end{align}
 
\begin{align}
    R_D\left(0,1\right)-R_D\left(0,0\right)\leq R_D\left(1,1\right)- R_D\left(1,0\right)
\end{align}

\begin{align}
    1-(1-C_D)\geq 0- (1-C_D)
\end{align}
\item Now consider the mapping 
T as defined in \eqref{T_map}. After performing the necessary calculations for our case, it is represented by the following matrices:
\begin{dgroup}
\begin{dmath}
T\left(.,b,0\right)= \pi\begin{bmatrix}
1 & 0  \\
1 & 0 \\
\end{bmatrix}
\end{dmath}

\begin{dmath}
T\left(1,b,1\right)=b P(0,1)\\
=b\begin{bmatrix}
e^{-\alpha} & 1-e^{-\alpha}  \\
 0 & 1 \\
\end{bmatrix}
\end{dmath}
\begin{dmath}
T\left(0,b,1\right)=b\left [  p(T_1/R_0) P(0,1)+p(T_0/R_0)P(0,0)\right ]\\
=\begin{bmatrix}
 \frac{(1-\nu)\pi_De^{-\alpha}+(1-\pi_D)}{(1-\nu)\pi_D+(1-\pi_D)} & \frac{(1-\nu)\pi_D(1-e^{-\alpha}) }{(1-\nu)\pi_D+(1-\pi_D)}  \\
0 & 1 \\
\end{bmatrix}
\end{dmath}
\end{dgroup}
From above, we conclude that,
\[T(b,y,u+1)\geq_r T(b,y,u)\quad \forall y \in \mathcal{Y} \; \forall b \in   \]
and if $y'> y$ then $T(b,y',u)\geq_r T(b,y,u)$
    \end{enumerate}
By Theorem 11.2.1 in \cite{krishnamurthy2016partially}, the value function \( V_{\pi}(b) \) is increasing in the belief \( b \). Furthermore, for a two-state Markov chain, as considered in our case, Theorem 11.3.1 in \cite{krishnamurthy2016partially} establishes that the optimal policy \( \pi \) follows a threshold structure.

\end{proof}
\begin{theorem}
Given a fixed defender policy \(\pi_D \in \Pi_D\), the attacker's value function \(V_{\pi_D}(s, b)\) is increasing in the state \(s\). Furthermore, the attacker's optimal policy exhibits a threshold structure: it is increasing in \(s\). Moreover, since the defender’s policy \(\pi_D\) is decreasing in the belief \(b\) (as established in the previous theorem), the attacker's optimal policy is increasing in \(b\).
\end{theorem}

\begin{proof}
    \begin{enumerate}
        \item First we show reward function is increasing super modular \begin{align}
    \begin{matrix}
R_A\left(0,0\right)=-C_A & R_A\left(0,1\right)= 0  \\
R_A\left(1,0\right)=1-C_A& R_A\left(1,1\right)=1 \\
\end{matrix}
\end{align}
\begin{align}
    R_A\left(0,1\right)-R_A\left(0,0\right)\leq R_A\left(1,1\right)- R_A\left(1,0\right)
\end{align}
Thus reward function is super-modular and increasing,

 \item Now consider the transition matrix for the attacker, where $u=d=0$ and  $u+1=d=1$ 

\begin{equation}
    P_{l}(u)=\begin{bmatrix}
1-(1-e^{-\alpha})(1-\pi_D)& (1-e^{-\alpha})(1-\pi_D) \\
 \pi_D & 1-\pi_D \\
\end{bmatrix}
\end{equation}
\begin{equation}
    P_{l}(u+1)=\begin{bmatrix}
1 &  0 \\
\pi_D  & 1-\pi_D \\
\end{bmatrix}
\end{equation}
From above matrix 
\begin{equation}
     P_s(u)\leq_s P_{s+1}(u) ,\forall u
\end{equation}

\item we show that transition probabilities $P_{ij}(u)$ are tail-sum supermodular
\begin{equation}
    \frac{1}{2}\left ( P_{s+1}(u+1)+ P_{i}(u)\right )\geq_s \frac{1}{2}\left ( P_{s}(u+1)+ P_{s+1}(u)\right )
\end{equation}
\begin{dmath}
\left [ \pi_D+ 1-(1-e^{-\alpha})(1-\pi_D) , 1-\pi_D+ (1-e^{-\alpha})(1-\pi_D) \right ]\geq_s \left [ 1+\pi_D,0+ 1-\pi_D \right ]
\end{dmath}
    \end{enumerate}

Given that the defender's policy $\pi_D$ is fixed and the three conditions mentioned earlier are satisfied, Theorem 9.3.1 implies that the value function increases with the state $s$. This means there exists an optimal policy that is also increasing in $s$. In particular, if it is optimal to take action $d=0$ (probe) when the state is $s=0$, then it is optimal to take action $d=1$ (not probe) when the state increases to $s=1$.
Moreover, since these properties hold with respect to the defender’s policy $\pi_D$, and if $\pi_D$ increases with the belief $b$ (as shown in the previous theorem), then there also exists an optimal policy that increases with $b$ and has a threshold structure. That is, there exists a threshold value $b_{\text{threshold}}$ such that it is optimal to stop probing when the belief exceeds this value.
\end{proof}
We assume static costs for the defender and attacker, but the threshold structure remains valid with dynamic costs if their net rewards maintain the same structure.
\section{Learning}\label{sec_learning}
In the previous section, we showed that the Nash equilibrium of a partially observable stochastic game follows a special policy structure, with thresholds based on the belief of the state . Building on this, we now aim to develop a reinforcement learning scheme that enables players to reach the Nash equilibrium \cite{policy_gradient_pano}.
To achieve this, we first approximate the threshold policy using a sigmoid function. Let \( b \) represent the current belief about the state and \( \theta \) be a threshold parameter. A threshold policy based on a sigmoid function is given by:  

\[
\pi(a | b, \theta) = \frac{1}{1 + e^{-K (b - \theta)}}
\]
Here, \( K \) is a scaling parameter that controls the steepness of the sigmoid function. A larger \( K \) results in a steeper transition around the threshold \( \theta \).
\begin{figure}
    \centering
    \resizebox{0.6\columnwidth}{0.45\columnwidth}{
   \begin{tikzpicture}
    \begin{axis}[
        axis lines=middle,
        xlabel={$x$}, ylabel={$y$},
        domain=0:1,
        samples=100,
        xmin=0, xmax=1,
        ymin=0, ymax=1,
        xtick={0,0.25,0.5,0.75,1},
        ytick={0,0.5,1},
    ]

\addplot[blue, thick] {1 / (1 + exp(-20*(x-0.45)))};
\addplot[red, thick] {1 / (1 + exp(-50*(x-0.45)))};
    \end{axis}
\end{tikzpicture}}

\caption{Approximation of the step function using a sigmoid function with a threshold of \(0.45\), shown for  \(k = 20\) (blue) and \(k = 50\) (red).}
\end{figure}
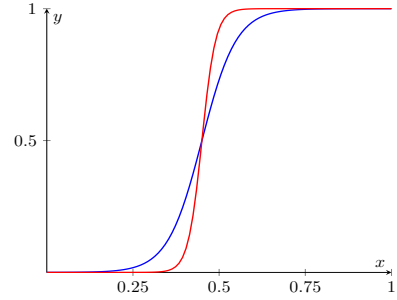 Consider the gradient of the objective function with respect to the parametric threshold. The corresponding policy gradient is given by: \cite{sutton1998reinforcement}\cite{sutton1999policy}
    \[
    \nabla J(\theta) \propto \mathbb{E}_{\pi}\left[\nabla \log \pi(a| b) Q(b, a)\right]
    \]
    Where:
    \begin{itemize}
      \item $\nabla J(\theta)$ is the gradient of the expected return with respect to the policy parameters $\theta$.
      \item $\pi(a| b)$ is the policy.
      \item $Q(b, a)$ is the action-value function.
    \end{itemize}
Next, we introduce our Policy Gradient-Based Fictitious Play for computing Nash equilibrium policies. In this method, the first player (the attacker) starts by fixing the opponent's (defender's) strategy and uses policy gradient techniques to determine its best response. Then, the second player fixes the first player's policy and applies policy gradient learning to find its best response. This iterative process continues until both players converge to a pair of policies that are mutual best responses.

\begin{algorithm}
\caption{Policy Gradient-Based Fictitious Play for Stochastic Game}

\begin{algorithmic}[1]
\State Initialize the threshold policy for the each player Defender and Attacker.
\Repeat~\textbf{}
\ForEach {Player $s\in \left\{ \text{Attacker, Defender} \right\} $}

    \State Consider the policy of the opponent as fixed.
    \State Collect trajectories using policies \( \pi_{\theta_s} \).
        \State Compute the returns \( R_s(\tau) \).
        \State Compute the policy gradients \( \nabla_{\theta_s} J_s(\theta_{s}, \theta_{-s}) \).
        \State Update the policy parameters:
        \[
        \theta_s \leftarrow \theta_s + \alpha_s \nabla_{\theta_s} J_s(\theta_{s}, \theta_{-s})
        \]
        \EndFor
\State Repeat steps 2 and 3 until convergence.
\Until {convergence}
\end{algorithmic}
\end{algorithm}

\section{Numerical results}\label{numerical}
\begin{figure*}[htbp!]
 \centering
\resizebox{0.3\textwidth}{!}
{
\begin{tikzpicture}
\node[anchor=south west] at (-1.3,5.5){\large\textbf{(a)}};
    \begin{axis}[
      xlabel={C_D},
      ylabel={Threshold},
      legend style={at={(0.5,-0.2)}, anchor=north,legend columns=-1},
      title={C_A=0.01}
    ]

   \addplot table [x index=0, y index=1, col sep=comma] {./data_files/50_alpha_20_gamma20_C_A_10.csv}; 
   \addplot table [x index=0, y index=2, col sep=comma] {./data_files/50_alpha_20_gamma20_C_A_10.csv};
    \addlegendentry{Defender}\addlegendentry{Attacker}\addlegendentry{Resource3}
       \end{axis}
  \end{tikzpicture}
 }  
\resizebox{0.3\textwidth}{!}
{

\begin{tikzpicture}
\node[anchor=south west] at (-1.3,5.5){\large\textbf{(b)}};
    \begin{axis}[
      xlabel={C_D},
      ylabel={Threshold},
      title={C_A=0.05},
      legend style={at={(0.5,-0.2)}, anchor=north,legend columns=-1},
    ]

   \addplot table [x index=0, y index=1, col sep=comma] {./data_files/50_alpha_10_gamma20_C_A_50.csv}; 
   \addplot table [x index=0, y index=2, col sep=comma] {./data_files/50_alpha_10_gamma20_C_A_50.csv};
    \addlegendentry{Defender}\addlegendentry{Attacker}\addlegendentry{Resource3}
       \end{axis}
  \end{tikzpicture}
 }  
\resizebox{0.3\textwidth}{!}
{
\begin{tikzpicture}
\node[anchor=south west] at (-1.3,5.5){\large\textbf{(c)}};
    \begin{axis}[
      xlabel={C_D},
      ylabel={Threshold},
      legend style={at={(0.5,-0.2)}, anchor=north,legend columns=-1},
      title={C_A=0.1}
    ]

   \addplot table [x index=0, y index=1, col sep=comma] {./data_files/50_alpha_10_gamma20_C_A_100.csv}; 
   \addplot table [x index=0, y index=2, col sep=comma] {./data_files/50_alpha_10_gamma20_C_A_100.csv};
    \addlegendentry{Defender}\addlegendentry{Attacker}\addlegendentry{Resource3}
       \end{axis}
  \end{tikzpicture}
 }  
\caption{The figure shows the threshold Nash equilibrium for various attacker and defender cost values.}
\end{figure*}
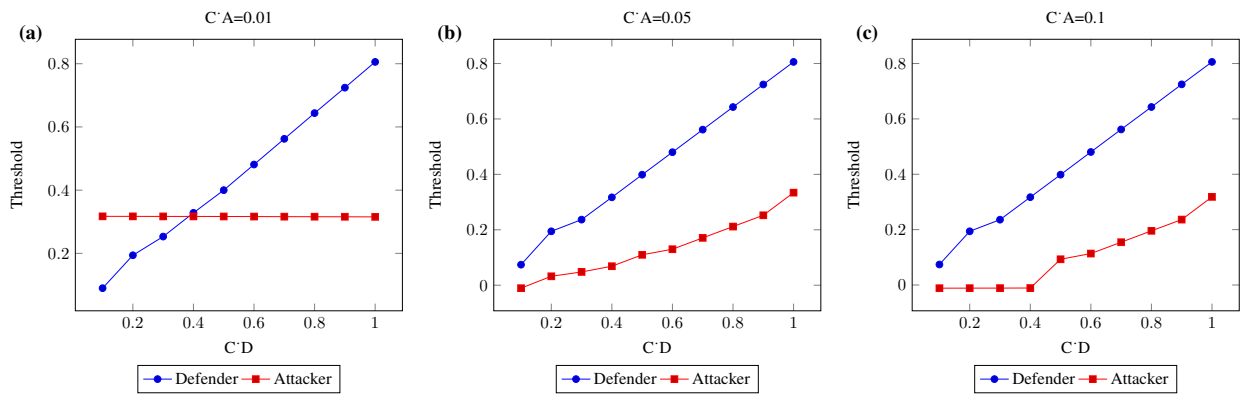
In this section, we validate the proposed learning algorithm through numerical simulations. We assume that the defender's reimaging cost ranges from 0.1 to 1, and we vary it between 0 and 1. Meanwhile, the attacker's probing cost is fixed at 0.01, 
0.05, and 0.1, respectively. We plot the Nash equilibrium threshold for the defender and attacker, for the defender it determines the belief above which the defender chooses to reimage the system, while for attacker it stop probing system.


We observe that, in the Fig.2 (a) attacker cost for each probe is $0.01$, the threshold for the defender increases with increasing cost of reimage, while for attacker it is independent of the cost of reimage and fix. For Fig. 2(b), we assume that each probe costs the attacker $0.05$. We observe that as the defender's reimaging cost increases, the Nash equilibrium thresholds for both players also increase. When the reimaging cost is very low, the defender can reimage the system frequently, reducing the attacker's chances of compromising it. As a result, the attacker's threshold to stop probing is zero, meaning they never initiate probes. However, as the cost of reimaging increases, the attacker has more opportunities to investigate and potentially compromise the system.

\par In Fig. 2(c), we consider a scenario where the attacker's probing cost is 0.1, which is higher than in the two previous cases. Initially, when the defender's reimaging cost is very low, the defender can frequently reimage the system, making the attacker's optimal threshold zero. Even as the reimaging cost increases, the threshold remains zero compared to the previous setting. This is because, although a higher reimaging cost would typically encourage the attacker to probe, the attacker's probing cost is also higher in this setting, preventing probing within a certain range of reimaging costs. However, when the reimaging cost increases further, it eventually becomes favorable for the attacker to probe and compromise the system.

\section{Conclusion}
This research explored the application of adaptive learning techniques in MTD to enhance cybersecurity strategies. By leveraging game theory and POMDP, we developed a dynamic framework where defenders can adjust their strategies against evolving cyber threats.
Our theoretical results demonstrate that the initial problem can be reduced to a threshold problem which is by nature much more scalable and is then a promising approach for an industrial application. Despite promising results, challenges such as computational complexity, deployment constraints, and the need for quality threat intelligence remain. Future research should focus on hybrid learning models, real-time threat detection, and adversarial AI to improve the effectiveness of MTD strategies in real-world cybersecurity.



\bibliography{aaai24}
\bibliographystyle{IEEEtran}
\end{document}